\documentclass[11pt]{article}

\usepackage[margin=1in]{geometry}
\usepackage[british]{babel}
\usepackage[utf8]{inputenc}
\usepackage{amsfonts,amsmath,amssymb,amsthm}
\usepackage{mathtools}
\usepackage{authblk}
\usepackage{hyperref}
\usepackage{cleveref}
\usepackage{thm-restate}
\usepackage{xcolor}
\usepackage{tikz}
\usetikzlibrary{decorations.pathreplacing}
\pdfoutput=1 

\newcommand{\odd}{\mathsf{odd}}
\newcommand{\OPT}{\mathrm{OPT}}
\renewcommand{\deg}{\mathop{\mathsf{deg}}\nolimits}

\theoremstyle{plain}
\newtheorem{theorem}{Theorem}
\newtheorem{lemma}[theorem]{Lemma}
\newtheorem{observation}[theorem]{Observation}
\theoremstyle{definition}
\newtheorem{problem}[theorem]{Problem}

\title{Improved Algorithms for Bounded-Degree (Subset) Traveling Salesman Problems\footnote{This work was partly supported by an IITP grant funded by the Korean Government (MSIT) (No. RS-2020-II201361, Artificial Intelligence Graduate School Program (Yonsei University)). This work was supported by the National Research Foundation of Korea(NRF) grant funded by the Korea government(MSIT) (RS-2025-00563707).}}

\author[1]{Jongseo Lee\thanks{\texttt{leejseo@kaist.ac.kr}}}
\author[2]{Jaehyeok Kwak\thanks{\texttt{dreami63@yonsei.ac.kr}}}
\author[2]{Hyung-Chan An\thanks{Corresponding Author. \texttt{hyung-chan.an@yonsei.ac.kr}}}
\affil[1]{School of Computing, KAIST, Daejeon, South Korea}
\affil[2]{Department of Computer Science and Engineering, Yonsei University, Seoul, South Korea}

\date{}

\begin{document}

\maketitle

\begin{abstract}
We present improved algorithms for several bounded-degree traveling salesman problems. In the bounded-degree traveling salesman path problem (BDTSPP), given a weighted graph $G=(V,E)$, two endpoints $s,t \in V$, and degree bounds $b_v$ for all $v \in V$, the goal is to find a minimum-cost subgraph of $G$ (possibly with some edges duplicated) that admits an Eulerian $s$-$t$ path and in which each vertex $v$ has degree at most $b_v$. Since deciding feasibility is already NP-hard for this problem, previous work gave a bicriteria approximation algorithm. However, that algorithm provides only a multiplicative guarantee on the degree violation, and it was left open whether a bicriteria approximation with additive degree violation is possible. We answer this open question affirmatively by giving a new bicriteria approximation algorithm with additive degree violation. The cost approximation ratio is improved as well, now matching that of Hoogeveen's analysis of the Christofides-Serdyukov algorithm.

This improvement relies on a new lemma that enables the use of a bounded-degree minimum spanning tree, rather than a bounded-degree Steiner tree, as a starting point for the algorithm. The lemma compares the cost and degrees of the tree against those of an integral optimum for the bounded-degree traveling salesman problem at hand, rather than those of a fractional optimum. Our lemma brings improvement to the circuit version (BDTSP) as well: we give a bicriteria algorithm that matches the previous cost approximation ratio while reducing the additive degree violation to $+2$, which is best possible.

Subset TSP is a generalization of the standard ``all-vertices'' TSP, in which only a specified subset of vertices is required to be visited. We present improvements for both the circuit and the path versions. For the subset path problem (BDSTSPP), we present the first bicriteria approximation algorithm with additive degree violation; for the subset circuit problem (BDSTSP), we give an improved cost approximation ratio.
\end{abstract}

\section{Introduction}

The traveling salesman problem (TSP) is one of the most fundamental problems in combinatorial optimization. For a long time, the best approximation ratios known for TSP were those achieved by the Christofides-Serdyukov algorithm~\cite{christofides,serdyukov1978nekotorykh}---$3/2$ for the circuit version and $5/3$ for the path version~\cite{hoogeveen1991}---and these bounds were improved only relatively recently to $(3/2 - 10^{-36})$~\cite{1.5-eps-random, traub2020reducing, karlin2023deterministic}. Despite this central importance, however, TSP remained largely unexplored in the bounded-degree setting for a long time, unlike other problems including minimum spanning tree, Steiner tree, or vertex/edge connectivity problems. Bounded-degree optimization is the natural problem of imposing degree constraints on solutions to network design problems. For TSP, this direction was initiated only recently by Friggstad and Mousavi~\cite{bdtsp}, who motivated the problem through an example of planning travel in a road network where certain locations should be traversed only a limited number of times, due to their disruptive nature or difficulty. We note that bounded-degree problems more generally arise in applications such as wireless sensor networks~\cite{amitabha2010bdmrst,an2025bdmrst}.

Since even deciding feasibility is already NP-hard for bounded-degree TSP, previous work has focused on bicriteria approximation algorithms. For the circuit version, Friggstad and Mousavi~\cite{bdtsp} gave a $(3/2,+4)$-approximation algorithm, i.e., an algorithm that returns a solution whose cost is at most $3/2$ times the optimum while allowing an additive degree violation of at most 4 at every vertex. This approximation ratio of $3/2$ matches that of the Christofides-Serdyukov algorithm and is therefore almost the best possible unless the approximation ratio for TSP without degree bounds itself is improved. In contrast, the additive degree violation of $+4$ still leaves room for improvement, and Friggstad and Mousavi~\cite{bdtsp} noted that it would be interesting to obtain an additive violation of $+2$. Progress on the path version has been much more limited. In this setting, Mousavi~\cite{mousavi-thesis} gave a bicriteria approximation algorithm whose degree guarantee is multiplicative: it outputs a solution of cost at most $8/3$ times the optimum and degree at most $\frac{5}{3}b_v + 4$ at each vertex $v$, where $b_v$ denotes the input degree bound of $v$. In the same work, Mousavi posed as an open question whether the problem admits a bicriteria approximation algorithm with an additive degree violation.

In this paper, we present improved algorithms for both the circuit and path problems. In both cases, our approximation ratios match those of the Christofides-Serdyukov algorithm. For the circuit problem, our algorithm achieves an additive degree violation of $+2$, which is best possible\footnote{Allowing $+1$ is no different from allowing $+0$ due to parity constraints.} since deciding feasibility is already NP-hard. For the path problem, we give the first algorithm whose degree violation is additive rather than multiplicative; the additive violation is $+4$.

Beyond these standard ``all-vertices'' problems, 
we study their subset generalizations.
In these problems, the input specifies a set of required vertices that must be visited, rather than requiring a visit to every vertex. 
Friggstad and Mousavi~\cite{bdtsp} introduced the bounded-degree subset traveling salesman problems
through the above example of planning travel in a road network, in which only a given subset of locations is now required to be visited.
Moreover, aside from such practical motivation, these problems may be of theoretical interest as they lie at the intersection of three major topics in combinatorial optimization: bounded-degree optimization, Steiner problems, and traveling salesman problems.
Friggstad and Mousavi gave $(5/3,+4)$-, $(13/8,+6)$-, and $(3/2,+8)$-approximation algorithms for the circuit version of this problem. As in the case of all-vertices problems, only a multiplicative degree guarantee was previously given for the path version: Mousavi~\cite{mousavi-thesis}'s algorithm outputs a solution of cost at most $8/3$ times the optimum and degree at most $\frac{5}{3}b_v + 4$ at each vertex $v$.

We give improvements for both problems in this paper. For the bounded-degree subset traveling salesman (circuit) problem, we present a $(14/9,+6)$-approximation algorithm; for the bounded-degree subset traveling salesman path problem, we give $(11/5,+6)$- and $(2,+8)$-approximation algorithms, achieving the first additive degree violation. 

\mbox{ }

A key ingredient in our results for the standard (all-vertices) problems is a simple new lemma that associates an optimal solution for bounded-degree TSP with a bounded-degree spanning tree. Most approximation algorithms for TSP ensure connectivity by incorporating a tree into the solution (and the degree parity is later corrected by adding a $T$-join). In previous works on bounded-degree TSP, this tree was obtained by computing bounded-degree Steiner trees rather than bounded-degree spanning trees, partly because the natural linear programming (LP) relaxation for bounded-degree TSP is formulated in terms of cut constraints, which are too weak as a spanning tree formulation. By comparing against a true integral optimal solution rather than an LP optimum, our lemma shows the existence of a bounded-degree spanning tree of low cost, as follows.

\begin{lemma}[Informal]\label{lem:informal-tree}
Given an instance of bounded-degree TSP with degree bounds $b_v$, let $\OPT$ denote its optimal cost. Then there exists a spanning tree of cost at most $\OPT$ in which every vertex $v$ has degree at most $\lfloor b_v/2\rfloor + 1$.
\end{lemma}

Using this lemma, we obtain the following improved performance guarantees.

\begin{theorem}\label{thm:bdtsp}
There exists a $(3/2,+2)$-approximation algorithm for the bounded-degree traveling salesman problem.
\end{theorem}

\begin{theorem}\label{thm:bdpathtsp}
There exists a $(5/3,+4)$-approximation algorithm for the bounded-degree traveling salesman path problem.
\end{theorem}

On the other hand, it is much more natural for the subset problems to use Steiner trees rather than spanning trees as the initial connected graph. For example, Friggstad and Mousavi~\cite{bdtsp} used a modified version of Lau and Singh's bounded-degree Steiner tree algorithm~\cite{bdstp} that imposes additional parity constraints on vertex degrees. While this parity-constrained solution is a 2-approximation solution to the problem, it may contain edges that are redundant for the purpose of terminal connectivity; as such, the algorithm uses its inclusion-wise minimal subset as the initial connected graph instead. Since we do not have any good bound on the cost of the redundant edges (which are absent from the minimal subset), the analysis just pays for the entire parity-constrained solution. Our analysis also pays for the entire parity-constrained solution, but by making heavier use of the redundant edges when we bound the cost of the $T$-join, we make sure these payments are not wasted, leading to the following improvements in the cost approximation ratios.

\begin{theorem}\label{thm:bd-subset-14-9-6}
There exists a $(14/9,+6)$-approximation algorithm for the bounded-degree subset traveling salesman problem.
\end{theorem}

\begin{theorem}\label{thm:bd-subset-path-11-5-6}
There exists a $(11/5,+6)$-approximation algorithm for the bounded-degree subset traveling salesman path problem.
\end{theorem}

We also give an alternative trade-off for the path problem.

\begin{theorem}\label{thm:bd-subset-path-2-8}
There exists a $(2,+8)$-approximation algorithm for the bounded-degree subset traveling salesman path problem.
\end{theorem}

Finally, as a complement to our result for the bounded-degree traveling salesman path problem, we present the following observation on the integrality gap.

\begin{observation}\label{obs:int_gap}
For every $\epsilon > 0$, there exists an instance of the bounded-degree traveling salesman path problem such that
\begin{itemize}
\item the cost of every (integral) solution, without any degree restrictions, is no smaller than $(3/2-\epsilon)$ times the LP optimum, and
\item every (integral) solution, without any cost restrictions, has a vertex $v$ whose degree is at least $b_v+2$.
\end{itemize}
\end{observation}

\paragraph*{Related work.}
Bounded-degree optimization settings have been studied for a variety of fundamental combinatorial optimization problems. For the minimum spanning tree problem, bicriteria approximation algorithms were studied by Ravi, Marathe, Ravi, Rosenkrantz, and Hunt~\cite{ravi1993many}, K\"{o}nemann and Ravi~\cite{konemann2002matter}, and K\"{o}nemann and Ravi~\cite{konemann2005primaldual}. While the degree violation is necessary since, otherwise, it is NP-hard even to find a feasible solution, algorithms that output a (super)optimal\footnote{Here, ``(super)optimal'' is with respect to the optimum for the original problem with no degree violation.} solution with a provable bound on degree violation were studied as well, e.g., Chaudhuri, Rao, Riesenfeld, and Talwar~\cite{chaudhuri2009edmonds}, Ravi and Singh~\cite{ravi2006delegate}, Goemans~\cite{goemans2006minimum}, and Singh \& Lau~\cite{dbst}. Steiner network problems were also extensively studied: we refer interested readers to, e.g., Ravi, Marathe, Ravi, Rosenkrantz, and Hunt~\cite{ravi2001degreeconstrained}, K\"{o}nemann and Ravi~\cite{konemann2003quasi}, Lau, Naor, Salavatipour, and Singh~\cite{lau2009survivable}, Louis and Vishnoi~\cite{louis2010improved}, Lau and Singh~\cite{bdstp}, and 
Lau \& Zhou~\cite{lau2015unified}. Vertex/edge connectivity problems have been studied under degree bounds as well~\cite{lau2009survivable,bansal2009additive,fukunaga2012iterative,nutov2012degree,khandekar2013network}.

Christofides~\cite{christofides} and Serdyukov~\cite{serdyukov1978nekotorykh} independently gave a $3/2$-approximation algorithm for the circuit TSP in the 1970s, and Hoogeveen~\cite{hoogeveen1991} showed that their algorithm is a $5/3$-approximation algorithm for the path TSP. These approximation ratios remained the best performance guarantees for both problems until relatively recently. For the circuit problem, Karlin, Klein, and Gharan~\cite{1.5-eps-random} gave the first constant improvement in the approximation ratio, by more than $10^{-36}$; this result was subsequently derandomized by Karlin, Klein, and Gharan~\cite{karlin2023deterministic}. For the path problem, An, Kleinberg, and Shmoys~\cite{pathtsp} gave the first constant improvement over $5/3$, followed by a sequence of further improvements by Seb\H{o}~\cite{sebo}, Vygen~\cite{vygen2016reassembling}, Gottschalk and Vygen~\cite{gottschalk2018better}, Seb\H{o} and van Zuylen~\cite{sebo2019salesman}, Traub and Vygen~\cite{traub2019approaching}, and Zenklusen~\cite{1.5}. This line of work culminated in a result of Traub, Vygen, and Zenklusen~\cite{traub2020reducing}, who showed a reduction from the path problem to the circuit version with an additive $\epsilon$ loss in the approximation ratio, for any fixed $\epsilon>0$.

\section{Preliminaries}\label{sec:prelim}
\subsection{Problems, Notation, and LP Relaxations}
Throughout this paper, we consider a connected undirected graph $G = (V, E)$,
nonnegative edge costs $\{c_e\}_{e \in E}$,
and degree bounds $\{b_v\}_{v \in V}$ where each $b_v$ is a strictly positive integer.
For $S\subsetneq V$ ($S\neq\emptyset$), let $\delta_G(S)$ (or $\delta(S)$ when clear from the context) denote the set of edges in the cut defined by $S$, i.e., $\delta_G(S):=\{\{i,j\}\in E(G)\mid |\{i,j\}\cap S|=1\}$.
For $x\in\mathbb{R}^E$ and $F\subseteq E$, $x(F)$ denotes $\sum_{f\in F}x_f$, and $c(x)$ denotes  $\sum_{e\in E}c_ex_e$.

We denote by $2G$ the multigraph obtained by duplicating every edge in $E$;
a subgraph of $2G$ may use each edge of $E$ at most twice. Given a (multi)graph $H=(V,E)$, let $\deg_H(v)$ for $v\in V$ denote the degree of $v$ in $H$. Let $\odd(H):=\{v\in V\mid \deg_H(v)\text{ is odd}\}$ denote the set of vertices that have odd degree in $H$.
For a (multi)set of edges $F$, the \emph{incidence vector} of $F$, denoted $\chi_F\in\mathbb{Z}_{\geq 0}^E$, is a vector of which the coordinate corresponding to $e\in E$ is the multiplicity of $e$ in $F$ for each $e\in E$. Let $c(F)$ denote the total edge cost of $F$.

Finally, let  $X\triangle Y:=(X\setminus Y)\cup (Y\setminus X)$ denote the symmetric difference of $X$ and $Y$.

\begin{problem}[Bounded-Degree Traveling Salesman Problem (BDTSP)]\label{prob:bdtsp}
Given a graph $G = (V, E)$, edge costs $\{c_e\}_{e \in E}$, and \emph{degree bounds} $\{b_v\}_{v \in V}$ with $b_v$ even for all $v \in V$,
find a connected spanning subgraph $H$ of $2G$ that minimizes $\sum_{e \in E(H)} c_e$ subject to:
\begin{itemize}
    \item $\deg_H(v)$ is even for all $v \in V$, and
    \item $\deg_H(v) \le b_v$ for all $v \in V$.
\end{itemize}
\end{problem}

Following is an LP relaxation for this problem.
\begin{equation}\label{eq:bd-tsp}
\begin{array}{ll@{}ll}
\text{minimize} \quad & \displaystyle \sum_{e \in E} c_e x_e \\
\text{subject to} \quad & x(\delta(S)) \ge 2, & & \forall S \subsetneq V,\; S \neq \emptyset, \\
& x(\delta(v)) \le b_v, & \quad & \forall v \in V, \\
& x_e \ge 0, & & \forall e \in E.
\end{array}
\end{equation}

\begin{problem}[Bounded-Degree Traveling Salesman Path Problem (BDTSPP)]\label{prob:bdpath}
Given a graph $G = (V, E)$, edge costs $\{c_e\}_{e \in E}$, two distinct vertices $s, t \in V$,
and \emph{degree bounds} $\{b_v\}_{v \in V}$ with $b_v$ odd if and only if $v \in \{s, t\}$,
find a connected spanning subgraph $H$ of $2G$ that minimizes $\sum_{e \in E(H)} c_e$ subject to:
\begin{itemize}
    \item $\deg_H(v)$ is odd if and only if $v \in \{s, t\}$, and
    \item $\deg_H(v) \le b_v$ for all $v \in V$.
\end{itemize}
\end{problem}

Following is an LP relaxation for this problem.
\begin{equation}\label{eq:bd-path-tsp}
\begin{array}{ll@{}ll}
\text{minimize} \quad & \displaystyle \sum_{e \in E} c_e x_e \\
\text{subject to} \quad & x(\delta(S)) \ge 1, & & \forall S \subsetneq V,\; S \neq \emptyset,\; |S \cap \{s, t\}| = 1, \\
& x(\delta(S)) \ge 2, & & \forall S \subsetneq V,\; S \neq \emptyset,\; |S \cap \{s, t\}| \ne 1, \\
& x(\delta(v)) \le b_v, & \quad & \forall v \in V, \\
& x_e \ge 0, & & \forall e \in E.
\end{array}
\end{equation}

\begin{problem}[Bounded-Degree Subset Traveling Salesman Problem (BDSTSP)]\label{prob:bdsubset}
Given a graph $G = (V, E)$, edge costs $\{c_e\}_{e \in E}$, a set of \emph{terminals} $X \subseteq V$,
and \emph{degree bounds} $\{b_v\}_{v \in V}$ with $b_v$ even for all $v \in V$,
find a connected subgraph $H$ of $2G$ that spans all terminals in $X$, minimizes $\sum_{e \in E(H)} c_e$, and subject to:
\begin{itemize}
    \item $\deg_H(v)$ is even for all $v \in V$, and
    \item $\deg_H(v) \le b_v$ for all $v \in V$.
\end{itemize}
\end{problem}

We use the following LP relaxation from \cite{bdtsp}.
\begin{equation}\label{eq:bd-subset-tsp}
\begin{array}{ll@{}ll}
\text{minimize} & \displaystyle \sum_{e \in E} c_e x_e \\
\text{subject to}
& x(\delta(S)) \ge 2, && \forall S \subseteq V,\; S \cap X \ne \emptyset,\; (V \setminus S) \cap X \ne \emptyset, \\
& x(\delta(v)) \le b_v, && \forall v \in V, \\
& x(\delta(v)) \le x(\delta(S)), && \forall S \subseteq V \setminus X,\; S \ne \emptyset,\; \forall v \in S, \\
& 0 \le x_e \le 2, && \forall e \in E.
\end{array}
\end{equation}

The vertices outside $X$ are called \emph{Steiner vertices}.
For $S\subsetneq V$ ($S\neq\emptyset$), we will use $S$ as shorthand for the cut defined by $S$ and its complement ($V\setminus S$). We say a cut $S$ is \emph{terminal-separating} if both $S \cap X$ and $(V \setminus S) \cap X$ are nonempty.

\begin{problem}[Bounded-Degree Subset Traveling Salesman Path Problem (BDSTSPP)]\label{prob:bdsubsetpath}
Given a graph $G = (V, E)$, edge costs $\{c_e\}_{e \in E}$, a set of \emph{terminals} $X \subseteq V$ with $s,t \in X$,
and \emph{degree bounds} $\{b_v\}_{v \in V}$ with $b_v$ odd if and only if $v \in \{s, t\}$,
find a connected subgraph $H$ of $2G$ that spans all terminals in $X$ and minimizes $\sum_{e \in E(H)} c_e$, subject to:
\begin{itemize}
    \item $\deg_H(v)$ is odd if and only if $v \in \{s, t\}$, and
    \item $\deg_H(v) \le b_v$ for all $v \in V$.
\end{itemize}
\end{problem}

\paragraph*{Bounded-degree spanning trees.}
In the \emph{bounded-degree spanning tree problem}, we are given a graph $G = (V, E)$, edge costs $\{c_e\}_{e \in E}$, and \emph{degree bounds} $\{B_v\}_{v \in V}$ as input,
and the goal is to find a minimum-cost spanning tree $F$ of $G$ that respects the degree bounds.

Singh and Lau~\cite{dbst} gave an iterative rounding algorithm for this problem.
\begin{lemma}[Singh and Lau~\cite{dbst}]\label{lem:bdst}
There is a polynomial-time algorithm that, for any feasible  instance, outputs a spanning tree with
$c(F)\le \OPT$ and $\deg_F(v)\le B_v+1$ for all $v\in V$, where
$\OPT$ is the cost of the optimal solution (with no degree violation).
\end{lemma}

\paragraph*{Bounded-degree Steiner trees.}
In the \emph{bounded-degree Steiner tree problem}, we are given a graph $G = (V, E)$, edge costs $\{c_e\}_{e \in E}$, a set of \emph{terminals} $X \subseteq V$, and \emph{degree bounds} $\{B_v\}_{v \in V}$ as input,
and the goal is to find a minimum-cost connected subgraph that spans all terminals while respecting the degree bounds. Consider the following LP relaxation for this problem.
\begin{equation}\label{eq:bdst}
\begin{array}{ll@{}ll}
\text{minimize} \quad & \displaystyle \sum_{e \in E} c_e x_e \\
\text{subject to} \quad & x(\delta(S)) \ge 1, & \quad & \forall S \subseteq V,\; S \cap X \ne \emptyset,\; (V \setminus S) \cap X \ne \emptyset, \\
& x(\delta(v)) \le B_v, & & \forall v \in V, \\
& x_e \ge 0, & & \forall e \in E.
\end{array}
\end{equation}

Lau and Singh~\cite{bdstp} gave a bicriteria approximation algorithm that, in polynomial time, finds a subgraph spanning all terminals, with degree at most $B_v+3$ at each vertex $v$, and with total edge cost at most twice the optimal value of~\eqref{eq:bdst}.
Friggstad and Mousavi~\cite{bdtsp} gave a modified iterative rounding procedure that enforces a prescribed set of Steiner vertices to have even degree:
\begin{lemma}[{\cite[Lemma~12]{bdtsp}}]\label{lem:fm-structured}
Consider an instance of the bounded-degree Steiner tree problem $(G=(V,E), \{c_e\}_{e \in E}, \{B_v\}_{v \in V}, X)$, together with a feasible solution $\bar x$ to~\eqref{eq:bdst} and a set $A \subseteq V \setminus X$ where $B_v = 1$ for all $v \in A$.
Then there exists a polynomial-time algorithm that outputs a connected subgraph (not necessarily a tree) $F$ of $2G$ spanning $X$ with the following properties:
\begin{itemize}
    \item $c(F) \le 2c(\bar{x})$;
    \item $\deg_{F}(v) \le B_v + 3$ for all $v \in V \setminus A$;
    \item $\deg_{F}(v) \le 8$ and $\deg_{F}(v)$ is even for all $v \in A$;
    \item 
    $\deg_{F^m}(v) \le B_v + 3$ for all $v \in V$, where 
    $F^m$ is an inclusion-wise minimal subset of edges of $F$ that is feasible for the bounded-degree Steiner tree instance.
\end{itemize}
\end{lemma}

\paragraph*{Bounded-degree $T$-joins.}
Given a graph $G=(V, E)$, edge costs $\{c_e\}_{e \in E}$, \emph{degree bounds} $\{b_v\}_{v \in V}$, and a vertex subset $T \subseteq V$ of even cardinality, the \emph{bounded-degree $T$-join problem} asks for a minimum-cost edge set $J \subseteq E$ such that, for all $v \in V$, $\deg_J(v)$ is odd if and only if $v \in T$, and $\deg_J(v) \le b_v$.

Following is an LP relaxation for the problem.
\begin{equation}\label{eq:bd-t-join}
\begin{array}{ll@{}ll}
\text{minimize} \quad & \displaystyle \sum_{e \in E} c_e x_e \\
\text{subject to} \quad & x(\delta(v)) \le b_v, & \quad &  \forall v \in V, \\
  & x(\delta(S)) \ge 1, & &\forall S \subsetneq V,\; |S \cap T| \text{ odd}, \\
 &x_e \ge 0, & & \forall e \in E.
\end{array}
\end{equation}

We say a cut $S \subsetneq V$ is  \emph{$T$-odd} if $|S \cap T|$ is odd.
The next lemma establishes the integrality of this relaxation; see~\cite[Theorem~36.8]{combopt-book} and~\cite[Theorem~8]{bdtsp}.
\begin{lemma}[Bounded-degree $T$-join integrality]\label{lem:bdtjoin_integrality}
Suppose that $b_v$ is odd if and only if $v \in T$. Then,~\eqref{eq:bd-t-join} is an integral LP, and we can find in polynomial time a minimum-cost bounded-degree $T$-join.
\end{lemma}
In particular, this lemma implies that the value of any feasible solution to \eqref{eq:bd-t-join} serves as an upper bound on the minimum cost of a bounded-degree $T$-join.

We conclude this section with the following lemma that is implicit in the proof of Theorem~3.4 in~\cite{pathtsp}.

\begin{lemma}[Path cut-feasibility]\label{lem:cut-feasibility-path}
Let $x$ be a feasible solution to~\eqref{eq:bd-path-tsp}
and  $F$ be a spanning tree of $G$.
Define $T \coloneqq \odd(F) \triangle \{s,t\}$.
Then, $y \coloneqq \frac{1}{3} x + \frac{1}{3} \chi_F$ satisfies $y(\delta(S)) \ge 1$ for all $T$-odd cut~$S$.
\end{lemma}
\section{Bounded-Degree Traveling Salesman Problems}\label{sec:tsp}

In this section, we present our results for the bounded-degree (all-vertices) traveling salesman problems.

\subsection{Existence of a Low-Degree Spanning Tree}\label{sec:tree}

The following lemma extracts a low-degree spanning tree from a feasible bounded-degree circuit or path.

\begin{lemma}\label{lem:euler-to-tree}
Let $H$ be a connected multigraph in which either every vertex has even degree, or exactly two vertices $s$ and $t$ ($s \ne t$) have odd degree.
Then there exists a spanning tree $F$ of $H$ such that
\[
    \deg_F(v) \le \left\lfloor \frac{\deg_H(v)}{2} \right\rfloor + 1 \quad \text{for all } v \in V.
\]
\end{lemma}

\begin{proof}
It suffices to consider the case where $H$ has exactly two odd-degree vertices $s$ and $t$: otherwise, the graph has an Eulerian circuit and we can delete an arbitrary edge to obtain a connected multigraph with exactly two odd-degree vertices. Note that any spanning tree of this modified graph is a spanning tree of the original multigraph $H$ as well, also satisfying the desired degree bounds.

Fix an Eulerian $s$-$t$ walk of $H$. We will construct a (rooted) spanning tree $F$ by traversing this walk and adding some of the edges on the walk to the tree. Let $s$ be the root node; initially, it is the only node of the partial tree. For each edge as we traverse the walk, if we reach a previously unvisited vertex for the first time, we add this edge to the rooted tree, adding the unvisited vertex as a leaf node. Note that we have a spanning tree by the end of this procedure.

Consider $\deg_F(v)$ for each vertex $v$.
Let us first consider the case where $v = s$. Note that the walk ``leaves'' $s$ exactly $\left\lceil \frac{\deg_H(v)}{2} \right\rceil$ times, and each of these departures can create at most one child of $s$ in $F$.
Hence $\deg_F(s) \le \left\lceil \frac{\deg_H(v)}{2} \right\rceil = \left\lfloor \frac{\deg_H(v)}{2} \right\rfloor + 1$.

Now consider $v \ne s$.
When the walk reaches $v$ for the first time, the entering edge connects $v$ to its parent in the tree.
After that, a new child of $v$ can be created only when the walk departs from $v$.
The number of departures from $v$ is $\left\lfloor \frac{\deg_H(v)}{2} \right\rfloor$, both when $v = t$ and when $v \notin \{s,t\}$.
Therefore $\deg_F(v) \le \left\lfloor \frac{\deg_H(v)}{2} \right\rfloor + 1$.
\end{proof}

The following lemma is immediate.
\begingroup
\renewcommand{\thelemma}{\ref{lem:informal-tree}}
\begin{lemma}[Rephrased]\label{lem:feasible-to-tree}
Given an instance $(G = (V, E), \{c_e\}_{e \in E}, \{b_v\}_{v \in V})$ of BDTSP or BDTSPP (with $s$ and $t$ additionally specified for BDTSPP), let $\OPT$ denote the optimal cost.
Then there exists a spanning tree $F \subseteq E$ with cost at most $\OPT$ such that every vertex $v$ has degree at most $\lfloor \frac{b_v}{2} \rfloor + 1$.
\end{lemma}
\addtocounter{lemma}{-1}
\endgroup

\subsection{Bounded-Degree Traveling Salesman Problem (BDTSP)}\label{sec:bdtsp}

Let $(G = (V, E), \{c_e\}_{e \in E}, \{b_v\}_{v \in V})$ be the instance.
Let $\mathbf{1}_{p}$ denote the indicator function for $p$.

\paragraph*{Algorithm description.}
Set the degree bounds to $B_v:=\frac{b_v}{2} + 1$ and apply \Cref{lem:bdst} to obtain a spanning tree $F$.
Let $T \coloneqq \odd(F)$.
For each $v \in V$, let $b'_v$ be the smallest integer satisfying $b'_v \ge \frac{b_v}{2}$ and $b'_v \equiv \mathbf{1}_{v \in T} \pmod{2}$.
Find a minimum-cost bounded-degree $T$-join $J$ with degree bounds $b'_v$ using \Cref{lem:bdtjoin_integrality} and output $H \coloneqq F \uplus J$.

\paragraph*{Algorithm analysis.}
\begingroup
\renewcommand{\thetheorem}{\ref{thm:bdtsp}}
\begin{theorem}
There exists a $(3/2,+2)$-approximation algorithm for the bounded-degree traveling salesman problem.
\end{theorem}
\addtocounter{theorem}{-1}
\endgroup

\begin{proof}
It suffices to show the algorithm's correctness only for the case where the input instance is feasible.
From \Cref{lem:feasible-to-tree}, \Cref{lem:bdst} returns a spanning tree $F$ with $c(F) \le \OPT$ and $\deg_F(v) \le \frac{b_v}{2} + 2$ for all $v \in V$.

Let $x^*$ be an optimal solution to~\eqref{eq:bd-tsp}.
The vector $\frac{x^*}{2}$ satisfies all cut constraints of~\eqref{eq:bd-t-join}, since every cut $S \subsetneq V$ ($S \neq \emptyset$) satisfies $x^*(\delta(S)) \ge 2$.
The vector also satisfies the degree constraints because $\frac{1}{2} x^*(\delta(v)) \le \frac{b_v}{2} \le b'_v$ for all $v \in V$.
From \Cref{lem:bdtjoin_integrality}, we have $c(J) \le  c(\frac{x^*}{2}) \le \frac{1}{2}\,\OPT$
and $\deg_J(v) \le b'_v \le \frac{b_v}{2} + 1$ for all $v \in V$.

Let $H \coloneqq F \uplus J$.
Then $H$ is connected, spans all vertices, and satisfies $\odd(H) = \odd(F) \triangle T = \emptyset$, which implies that $H$ is  Eulerian.
Its cost is bounded by
\[
c(H) = c(F) + c(J) \le \OPT + \frac{1}{2}\,\OPT = \frac{3}{2}\,\OPT.
\]
Moreover,
\[
\deg_H(v) \le \deg_F(v) + \deg_J(v) \le \left(\frac{b_v}{2} + 2\right) + \left(\frac{b_v}{2} + 1\right) = b_v + 3.
\]
Since $\deg_H(v)$ and $b_v$ have the same parity and $\deg_H(v) \le b_v+3$, we have $\deg_H(v) \le b_v+2$.
\end{proof}

\subsection{Bounded-Degree Traveling Salesman Path Problem (BDTSPP)}\label{sec:bdptsp}

Consider an instance $(G = (V, E), \{c_e\}_{e \in E}, \{b_v\}_{v \in V}), s, t)$ of BDTSPP.

\paragraph*{Algorithm description.}
Set the degree bounds to $B_v:= \left\lfloor \frac{b_v}{2} \right\rfloor + 1$ and apply \Cref{lem:bdst} to obtain a spanning tree $F$.
Let $T \coloneqq \odd(F) \triangle \{s,t\}$.
For each $v \in V$, let $b'_v$ be the smallest integer satisfying $b'_v \ge \frac{b_v}{2} + \frac{2}{3}$ and $b'_v \equiv \mathbf{1}_{v \in T} \pmod{2}$.
Finally, find a minimum-cost bounded-degree $T$-join $J$ with degree bounds $b'_v$ using \Cref{lem:bdtjoin_integrality} and output $H := F \uplus J$.

The analysis of this algorithm proceeds similarly to the proof of Theorem~\ref{thm:bdtsp}, except that it uses Lemma~\ref{lem:cut-feasibility-path}. We defer the details to Appendix~\ref{app:534}.

\subsection{Integrality Gap Instances}\label{sec:path-gap}
The details of our integrality gap instance showing Observation~\ref{obs:int_gap} are deferred to Appendix~\ref{app:intgap}. This gap also extends to the subset problem by taking all vertices as terminals.

\begingroup
\renewcommand{\theobservation}{\ref{obs:int_gap}}
\begin{observation}
For every $\epsilon > 0$, there exists an instance of the bounded-degree traveling salesman path problem such that
\begin{itemize}
\item the cost of every (integral) solution, without any degree restrictions, is no smaller than $(3/2-\epsilon)$ times the LP optimum, and
\item every (integral) solution, without any cost restrictions, has a vertex $v$ whose degree is at least $b_v+2$.
\end{itemize}
\end{observation}
\addtocounter{observation}{-1}
\endgroup

\section{Bounded-Degree Subset Traveling Salesman Problems}\label{sec:subsettsp}

\subsection{Bounded-Degree Subset Traveling Salesman Problem (BDSTSP)}\label{sec:subset}
The analysis of our $(14/9,+6)$-approximation algorithm for BDSTSP uses ingredients already present in the analysis of our algorithm for the path version. We defer the detailed presentation of our algorithm for BDSTSP to Appendix~\ref{app:subset:circuit}.

\subsection{Bounded-Degree Subset Traveling Salesman Path Problem (BDSTSPP)}\label{sec:subset-path}

Consider an instance $(G = (V, E), \{c_e \}_{e \in E}, \{b_v\}_{v \in V}, X, s, t)$ of BDSTSPP.
We use the following LP relaxation.
\begin{equation}\label{eq:bd-subset-path-tsp}
\hspace{-1.7ex}\begin{array}{ll@{}ll}
\text{minimize} & \displaystyle \sum_{e \in E} c_e x_e \\
\text{subject to}
& x(\delta(S)) \ge 1, && \forall S \subseteq V,\; S \cap X \ne \emptyset,\; (V \setminus S) \cap X \ne \emptyset,\; |S \cap \{s,t\}| = 1, \\
& x(\delta(S)) \ge 2, && \forall S \subseteq V,\; S \cap X \ne \emptyset,\; (V \setminus S) \cap X \ne \emptyset,\; |S \cap \{s,t\}| \ne 1, \\
& x(\delta(v)) \le b_v ,&& \forall v \in V, \\
& x(\delta(v)) \le x(\delta(S)), && \forall S \subseteq V \setminus X,\; S \ne \emptyset,\; \forall v \in S, \\
& 0 \le x_e \le 2, && \forall e \in E.
\end{array}
\end{equation}

The validity of the fourth set of constraints is verified by a straightforward adaptation of the proof of \cite[Lemma~9]{bdtsp} to the path setting: if an integral solution violates this constraint for some $v$ and $S$, the pigeonhole principle implies that one can find a cycle containing $v$ that is completely within $S$, which can be deleted without losing optimality. This shows that \eqref{eq:bd-subset-path-tsp}  indeed is a relaxation.

\subsubsection{A \texorpdfstring{$(11/5, +6)$}{(11/5, +6)}-Approximation Algorithm}\label{sec:511}

\paragraph*{Algorithm description.}
Solve~\eqref{eq:bd-subset-path-tsp} to obtain an optimal solution $x^*$.
Let $A := \{v \in V \setminus X : x^*(\delta(v)) < 2\}$, and let
\[
B_v := \begin{cases}
    1, & v\in A, \\
    \frac{b_v + 1}{2}, & v \in \{s, t \}, \\
    \frac{b_v}{2}, & v \in V \setminus (A \cup \{s, t\}).
\end{cases}
\]
Add a zero-cost\footnote{Naturally, we will delete this auxiliary edge later; therefore, the choice of its cost is immaterial but setting it to zero simplifies the presentation.} auxiliary edge $e_0 = st$ and let $G_0 :=  (V,\,E \cup \{e_0\})$.
Let
\[
\bar x := \frac{x^* + \chi_{\{e_0\}}}{2}.
\]

Apply \Cref{lem:fm-structured} to the bounded-degree Steiner tree instance $(G_0, \{c_e\}_{e \in E}, \{B_v\}_{v \in V},X)$ with $\bar{x}$ and $A$.
This gives a connected subgraph $F_0 \subseteq 2G_0$ spanning $X$, and an inclusion-wise minimal subgraph
$F_0^m \subseteq F_0$ that is still feasible for the same bounded-degree Steiner tree instance.

Let $F^m$ be the subgraph of the original graph $G$ obtained by deleting all copies of $e_0$ from $F_0^m$.
Define $T := \odd(F^m) \triangle \{s,t \}$.
For each vertex $v$, set the ``residual'' degree bound
\[
r_v := b_v + 6 - \deg_{F^m}(v),
\]
and find a minimum-cost bounded-degree $T$-join $J$ of $G$ with degree bounds $r_v$ using \Cref{lem:bdtjoin_integrality}.

Output $H := F^m \uplus J$, after deleting any connected component that contains no terminal.

\paragraph*{Algorithm analysis.}
Let $F$ be the subgraph of $G$ obtained by deleting all copies of $e_0$ from $F_0$. Note that $F^m \subseteq F$. Let $F^r := F \setminus F^m$ be the multiset of ``remaining'' edges.

Note that one of the following two things happens: either $F^m$ has a connected component containing all terminals (as $F_0^m$ does), or $F^m$ has exactly two components containing terminal vertices: one containing $s$ and the other containing $t$ in particular.
In the latter case, we denote the vertex sets of the two connected components by $C_s$ and $C_t$.

Let $P$ be a shortest simple $s$-$t$ path in $G$. If $F^m$ has a component containing all terminals, set $Q:=\emptyset$.
Otherwise, $P$ is a path from $s\in C_s$ to $t\in C_t$; choose a minimal subpath $Q$ of $P$ with one endpoint in $C_s$ and the other endpoint in $C_t$.
Then no internal vertex of $Q$ lies in $C_s\cup C_t$. We therefore have $\deg_Q(v)\le 1$ for all $v\in C_s\cup C_t$, and $\deg_Q(v)\le2$ for all $v\in V$.
Moreover, $Q$ crosses every $s$-$t$ cut $S$ such that $\delta_{F^m}(S)=\emptyset$ and $c(Q)\le c(P)\le \OPT$ since every BDSTSPP solution contains an $s$-$t$ path as a subgraph.

\begin{lemma}\label{lem:bdstspp-join-integral}
The residual degree bound $r_v$ is odd if and only if $v \in T$, making $(G, \{c_e\}_{e \in E},$ $ \{r_v \}_{v \in V}, T)$ a valid bounded-degree $T$-join instance.
The vector
\[
y := \frac{1}{5} \chi_{F^m} + \frac{2}{5} \chi_{F^r} + \frac{2}{5} x^* + \frac{1}{5} \chi_P + \frac{2}{5}\chi_Q
\]
is a feasible solution to~\eqref{eq:bd-t-join} for this instance.
\end{lemma}

\begin{proof}
First, we show that $r_v$ is odd if and only if $v \in T$.
For $v \in \{s,t\}$, the bound $b_v$ is odd, so $r_v=b_v+6-\deg_{F^m}(v)$ is odd if and only if $\deg_{F^m}(v)$ is even.
This is exactly the condition $v\in \odd(F^m)\triangle\{s,t\}=T$.
For $v\notin\{s,t\}$, the bound $b_v$ is even, so $r_v$ is odd if and only if $\deg_{F^m}(v)$ is odd, again exactly when $v\in T$.

Now we verify the cut and degree inequalities for $y$. Consider a \(T\)-odd cut \(\delta(S)\), i.e., \(|S\cap T|\) is odd.
Since \(|T|\) is even, the complementary side \(V\setminus S\) is also \(T\)-odd. Thus, by taking the appropriate side of the cut, we can assume without loss of generality that  either \(\delta(S)\) is terminal-separating or \(S\subseteq V\setminus X\).

\begin{itemize}
    \item Case 1: $S \cap X \neq \emptyset$, $(V \setminus S) \cap X \neq \emptyset$, and $\delta(S)$ is not an $s$-$t$ cut.

    In this case, $x^*(\delta(S)) \ge 2$. We claim that  $F^m$ crosses $\delta(S)$: if $F^m$ connects all terminals, this is immediate.
    Otherwise, in order for $F^m \cap \delta(S)$ to be empty while $S$ is terminal-separating, one of $C_s$ and $C_t$ must be completely inside $S$ and the other must be disjoint from $S$. This, however, contradicts the assumption that $\delta(S)$ is not an $s$-$t$ cut. Hence,
    \[
    y(\delta(S)) \ge \frac{1}{5} \chi_{F^m}(\delta(S)) + \frac{2}{5} x^*(\delta(S)) \ge \frac{1}{5} \cdot 1 + \frac{2} {5} \cdot 2 = 1. 
    \]

    \item Case 2: $S \cap X \neq \emptyset$, $(V \setminus S) \cap X \neq \emptyset$, $\delta(S)$ is an $s$-$t$ cut, and $F^m$ crosses $\delta(S)$.

    We have $x^*(\delta(S)) \ge 1$. Note that $P$ crosses $\delta(S)$.    
    Since $S$ is $T$-odd and $|S\cap\{s,t\}|=1$, $|S\cap \odd(F^m)| \equiv |S\cap T| + |S\cap\{s,t\}| \equiv 0 \pmod 2$.
    By the handshaking lemma, $\chi_{F^m}(\delta(S)) \equiv |S\cap \odd(F^m)| \equiv 0 \pmod 2$.
    Since $\chi_{F^m}(\delta(S))>0$, we have $\chi_{F^m}(\delta(S))\ge2$. Therefore we have
    \[
    y(\delta(S)) \ge \frac{1}{5} \chi_{F^m}(\delta(S)) + \frac{2}{5} x^*(\delta(S)) + \frac{1}{5} \chi_P(\delta(S))  \ge \frac{1}{5} \cdot 2 + \frac{2}{5} \cdot 1 + \frac{1}{5} \cdot 1 = 1.
    \]

    \item Case 3: $S \cap X \neq \emptyset$, $(V \setminus S) \cap X \neq \emptyset$, $\delta(S)$ is an $s$-$t$ cut, and $F^m$ does not cross $\delta(S)$.

    This case can occur only if $F^m$ contains two components.
    In this case, $x^*(\delta(S)) \ge 1$, and $P$ and $Q$ cross $\delta(S)$.
    Hence,
    \[
    y(\delta(S)) \ge \frac{2}{5} x^*(\delta(S)) +  \frac{1}{5} \chi_P(\delta(S)) + \frac{2}{5} \chi_Q(\delta(S)) \ge \frac{2}{5} \cdot 1 + \frac{1}{5} \cdot 1 + \frac{2}{5} \cdot 1 = 1.
    \]

    \item Case 4: $S \subseteq V \setminus X$ and $S \setminus A \neq \emptyset$.

        Let $v \in S \setminus A$.
    Then, $x^*(\delta(S)) \ge x^*(\delta(v)) \ge 2$ due to the fourth constraint of~\eqref{eq:bd-subset-path-tsp}.
    Since $S\cap\{s,t\}=\emptyset$ and $S$ is $T$-odd, we have $|S\cap \odd(F^m)| \equiv |S\cap T| \equiv 1 \pmod 2$.
    By the handshaking lemma, $\chi_{F^m}(\delta(S))$ is odd, and hence
    $\chi_{F^m}(\delta(S))\ge1$.
    Thus,
    \[
    y(\delta(S)) \ge \frac{1}{5} \chi_{F^m}(\delta(S)) +\frac{2}{5} x^*(\delta(S)) \ge \frac{1}{5} \cdot 1 + \frac{2}{5} \cdot 2 = 1.
    \]

    \item Case 5: $S \subseteq V \setminus X$ and $S \subseteq A$.
    
    Since $S\cap\{s,t\}=\emptyset$ and $S$ is $T$-odd, the same parity argument as in Case 4 shows that $\chi_{F^m}(\delta(S))$ is odd.
    Suppose towards contradiction that $\chi_{F^m}(\delta(S)) = 1$. Then, since $e_0$ is not incident to $S$, the subgraph $F_0^m$ has exactly one edge crossing $\delta(S)$.
    Since $S$ contains no terminal, removing all edges of $F_0^m$ incident to vertices in $S$ preserves connectivity among terminals, contradicting the inclusion-wise minimality of $F_0^m$.
    Therefore, $\chi_{F^m}(\delta(S)) \ge 3$.
    Also, every vertex of $A$ has even degree in $F_0$ and $e_0$ is not incident to any vertex in $A$.
    Hence, $\chi_{F}(\delta(S))=\chi_{F_0}(\delta(S))$ is even.
    Since $F = F^m \uplus F^r$ and $\chi_{F^m}(\delta(S))$ is odd, $\chi_{F^r}(\delta(S))$ is odd. This implies $\chi_{F^r}(\delta(S)) \ge 1$.
    Consequently,
    \[
    y(\delta(S)) \ge \frac{1}{5} \cdot \chi_{F^m} (\delta(S)) + \frac{2}{5} \chi_{F^r} (\delta(S)) \ge \frac{1}{5} \cdot 3 + \frac{2}{5} \cdot 1  =  1.
    \]
\end{itemize}
This proves all $T$-odd cut constraints.

It remains to check the degree inequalities. Since $r_v := b_v + 6 - \deg_{F^m}(v)$, it suffices to check whether $\deg_{F^m}(v) + y(\delta(v))  \le b_v + 6$ holds for all $v$.

For every vertex $v$,
\[
\deg_{F^m}(v)+y(\delta(v))
=
\frac45\deg_{F^m}(v)
+\frac25\deg_F(v)
+\frac25x^*(\delta(v))
+\frac15\deg_P(v)
+\frac25\deg_Q(v)
\]
holds since $\deg_{F}(v)=\deg_{F^m}(v)+\deg_{F^r}(v)$.

First consider $v\in V\setminus(A\cup\{s,t\})$.
Then $\deg_F(v)\le b_v/2+3$, $\deg_{F^m}(v)\le\deg_F(v)$,
$x^*(\delta(v))\le b_v$, $\deg_P(v)\le2$, and $\deg_Q(v)\le2$.
Therefore,
\[
\deg_{F^m}(v)+y(\delta(v))
\le
\left(\frac45+\frac25\right)\cdot\left(\frac{b_v}{2}+3\right)
+\frac25 b_v
+\frac15\cdot2
+\frac25\cdot2
=
b_v+\frac{24}{5}
< b_v+6.
\]

Next consider $v \in A$.
Note that $b_v\ge2$.
We have $\deg_F(v)\le8$, $\deg_{F^m}(v)\le B_v+3=4$,
$x^*(\delta(v))<2$, and $\deg_P(v)\le2$.
Assume for now that $\deg_{F^m}(v)=0$. Since $\deg_Q(v)\le 2$, we have
\[
\deg_{F^m}(v)+y(\delta(v))
<
\frac25\cdot8+\frac25\cdot2+\frac15\cdot2+\frac25\cdot2
=
\frac{26}{5}
\le b_v+6.
\]
Otherwise, $\deg_{F^m}(v)>0$. The minimality of $F_0^m$ implies that $v$ belongs to a connected component of $F^m$ that contains terminals. If $F^m$ connects all terminals, we have $\deg_Q(v)=0$. Otherwise, we have $v\in C_s\cup C_t$ and therefore $\deg_Q(v)\le1$. We have in both cases
\[
\deg_{F^m}(v)+y(\delta(v))
<
\frac45\cdot4
+\frac25\cdot8
+\frac25\cdot2
+\frac15\cdot2
+\frac25\cdot1
=
8
\le b_v+6.
\]

Finally, consider $v\in\{s,t\}$.
We have $\deg_F(v)\le B_v+3 = (b_v+1)/2+3$, $\deg_{F^m}(v)\le \deg_F(v)$,
$x^*(\delta(v))\le b_v$, $\deg_P(v)=1$, and $\deg_Q(v)\le1$.
We thus have
\[
\deg_{F^m}(v)+y(\delta(v))
\le
\left(\frac45+\frac25\right)\cdot\left(\frac{b_v+1}{2}+3\right)
+\frac25 b_v
+\frac15
+\frac25
=
b_v+\frac{24}{5}
< b_v+6.
\]

This shows that  $y$ is a feasible 
solution to~\eqref{eq:bd-t-join}.
\end{proof}

\begingroup
\renewcommand{\thetheorem}{\ref{thm:bd-subset-path-11-5-6}}
\begin{theorem}
There exists a $(11/5,+6)$-approximation algorithm for the bounded-degree subset traveling salesman path problem.
\end{theorem}
\addtocounter{theorem}{-1}
\endgroup

\begin{proof}
First, we claim that $\bar x$ is feasible for~\eqref{eq:bdst} on $G_0$ with the set of terminals $X$ and degree bounds $\{B_v\}_{v \in V}$.
If $S$ is a terminal-separating cut and $|S\cap\{s,t\}|=1$, then
\[
\bar x(\delta_{G_0}(S))=\frac{x^*(\delta_G(S))+1}{2}\ge1.
\]
If $S$ is a terminal-separating cut and $|S\cap\{s,t\}|\neq1$, then
\[
\bar x(\delta_{G_0}(S))=\frac{x^*(\delta_G(S))}{2}\ge1.
\]
The degree inequalities follow from the definition of $\{B_v\}_{v \in V}$: for $v\in A$,
$\bar x(\delta_{G_0}(v))=\frac{x^*(\delta_G(S))}{2}<1=B_v$; for $v\in\{s,t\}$,
$\bar x(\delta_{G_0}(v))\le (b_v+1)/2=B_v$; and for all remaining vertices,
$\bar x(\delta_{G_0}(v))\le b_v/2=B_v$.

Recall that $e_0$ has zero cost. From \Cref{lem:fm-structured}, we have
\[
c(F)=c(F_0)\le 2c(\bar x)=c(x^*)\le \OPT.
\]
We also have $c(F^m)+c(F^r)=c(F)$, and $c(Q)\leq c(P)\le\OPT$.

Lemma~\ref{lem:bdstspp-join-integral} implies that the $T$-join $J$ found by the algorithm satisfies
$c(J)\le c(y)$ and $\deg_J(v)\le r_v$ for all $v\in V$.
Thus,
\begin{align*}
c(H)
&\le c(F^m)+c(y)\\
&=
\frac65 c(F^m)
+\frac25 c(F^r)
+\frac25 c(x^*)
+\frac15 c(P)
+\frac25 c(Q)\\
&\le
\frac65 c(F)
+\frac25 c(x^*)
+\frac15 c(P)
+\frac25 c(Q)\\
&\le
\left(\frac65+\frac25+\frac15+\frac25\right)\OPT
=
\frac{11}{5}\OPT.
\end{align*}

The degree bound follows from the definition of $r_v$: $\deg_H(v)
\le
\deg_{F^m}(v)+\deg_J(v)
\le
\deg_{F^m}(v)+r_v
=
b_v+6$.

Moreover, $\odd(F^m\uplus J)=\odd(F^m)\triangle T=\{s,t\}$.
All terminals are connected in $F^m$, except that $F^m$ may have two terminal-containing components, one containing $s$ and one containing $t$.
Since every connected component of a graph has an even number of odd-degree vertices, $s$ and $t$ lie in the same connected component of $F^m\uplus J$.
Hence, all terminals lie in the same connected component.
Deleting any component that contains no terminal does not affect feasibility.
\end{proof}

\subsubsection{A \texorpdfstring{$(2, +8)$}{(2, +8)}-Approximation Algorithm}
We provide an alternative trade-off of the bicriteria performance guarantee;
although the algorithm and its analysis reuse those from Section~\ref{sec:511}, Appendix~\ref{app:bdstspp:2:8} gives a detailed presentation of the alternative algorithm for the sake of completeness.
	
\bibliography{cite}

\appendix

\section{Deferred from Section~\ref{sec:tsp}: Bounded-Degree Traveling Salesman Problems}

\subsection{Analysis of Our \texorpdfstring{$(5/3, +4)$}{(5/3, +4)}-Approximation Algorithm for BDTSPP}\label{app:534}

\begingroup
\renewcommand{\thetheorem}{\ref{thm:bdpathtsp}}
\begin{theorem}
There exists a $(5/3,+4)$-approximation algorithm for the bounded-degree traveling salesman path problem.
\end{theorem}
\addtocounter{theorem}{-1}
\endgroup

\begin{proof}
It suffices to show the algorithm's correctness only for the case where the input instance is feasible.  From \Cref{lem:feasible-to-tree}, \Cref{lem:bdst} returns a spanning tree $F$ with $c(F) \le \OPT$ and  $\deg_F(v) \le \left\lfloor \frac{b_v}{2} \right\rfloor + 2$ for all $v \in V$.

Let $x^*$ be the optimal solution to~\eqref{eq:bd-path-tsp} and
$y \coloneqq \frac{1}{3} x^* + \frac{1}{3} \chi_F$.
By \Cref{lem:cut-feasibility-path}, $y$ satisfies all $T$-odd cut constraints of~\eqref{eq:bd-t-join}.
For all $v \in V$,
\[
y(\delta(v)) = \frac{1}{3}\bigl(x^*(\delta(v)) + \deg_F(v)\bigr)
\le \frac{1}{3}\left(b_v + \left\lfloor \frac{b_v}{2} \right\rfloor + 2\right)
\le \frac{b_v}{2} + \frac{2}{3}.
\]
Hence, $y$ is feasible for~\eqref{eq:bd-t-join} with degree bounds $b'$. Moreover, the bounds $b'$ have the parity required by \Cref{lem:bdtjoin_integrality}. We thus have
\[
c(J) \le c(y) = \frac{1}{3} c(x^*) + \frac{1}{3} c(F) \le \frac{2}{3}\,\OPT
\]
and $\deg_J(v) \le b'_v < \frac{b_v}{2} + \frac{2}{3} + 2$ for all $v$.

Let $H \coloneqq F \uplus J$.
Then $H$ is connected, spans all vertices, and satisfies $\odd(H) = \odd(F) \triangle T = \{s,t\}$.
Its cost is bounded by
\[
c(H) = c(F) + c(J) \le \OPT + \frac{2}{3}\,\OPT = \frac{5}{3}\,\OPT.
\]
Moreover,
\[
\deg_H(v) \le \deg_F(v) + \deg_J(v) < \left(\frac{b_v}{2} + 2\right) + \left(\frac{b_v}{2} + \frac{2}{3} + 2\right) = b_v + \frac{14}{3}.
\]
Since $\deg_H(v)$ and $b_v$ have the same parity and $\deg_H(v) < b_v + \frac{14}{3}$, we get $\deg_H(v) \le b_v+4$.
\end{proof}

\subsection{Integrality Gap Instances}\label{app:intgap}
We present a family of cycle instances showing the following.
\begingroup
\renewcommand{\theobservation}{\ref{obs:int_gap}}
\begin{observation}[Rephrased]
For all $\epsilon>0$, there exists an instance of the bounded-degree traveling salesman path problem such that~\eqref{eq:bd-path-tsp} is feasible, but every integral solution violates some degree bound by at least $2$, and, even after allowing  arbitrary  degree violation, the ratio between the minimum integral cost and the LP optimum is at least $3/2-\epsilon$.
\end{observation}
\addtocounter{observation}{-1}
\endgroup

For $n\geq 4$, let $G = C_{2n}$ with unit edge costs, and choose $s$ and $t$ as antipodal vertices.
Set $b_s = b_t = 3$ and $b_v = 2$ for $v \notin \{s, t\}$.
The vector $x_e = 1$ for all edges of the cycle is feasible for~\eqref{eq:bd-path-tsp} (and for~\eqref{eq:bd-subset-path-tsp} with $X \coloneqq V$) with cost $2n$.

There is no feasible integral solution that exactly satisfies all degree bounds.
Any integral solution must violate at least one degree bound, and by parity the smallest possible additive violation is $+2$.
Even if we allow arbitrary degree violation, the minimum integral cost is $3n-2$. Consequently, the ratio between the minimum cost of an integral solution with no degree restrictions and the LP optimum is at least
\[
\frac{3n - 2}{2n} = \frac{3}{2} - \frac{1}{n},
\]
which tends to $\frac{3}{2}$ as $n \to \infty$.

\begin{figure}[htbp]
    \centering
    \begin{tikzpicture}[scale=1.5, every node/.style={circle, draw, fill=white, minimum size=0.6cm, inner sep=0pt}]
        \def\radius{1.5}
        \foreach \i in {0,...,7} {
            \pgfmathsetmacro{\angle}{90 - \i * 360 / 8}
            \node (v\i) at (\angle:\radius) {};
        }
        \node[fill=gray!20] at (90:\radius) {$s$};
        \node[fill=gray!20] at (-90:\radius) {$t$};
        \draw[thick] (v0) -- (v1) -- (v2) -- (v3) -- (v4) -- (v5) -- (v6) -- (v7) -- (v0);
        \node[draw=none, fill=none] at (90:\radius+0.4) {$b_s=3$};
        \node[draw=none, fill=none] at (-90:\radius+0.4) {$b_t=3$};
        \node[draw=none, fill=none] at (0:\radius+0.6) {$b_v=2$};
    \end{tikzpicture}
    \caption{The cycle instance for $n=4$. The all-ones vector on the cycle has LP cost $2n$, but no exact bounded-degree integral $s$-$t$ path  exists. Even if arbitrary degree violations are allowed, the minimum integral solution cost is $3n-2$.}
    \label{fig:gap_instance}
\end{figure}
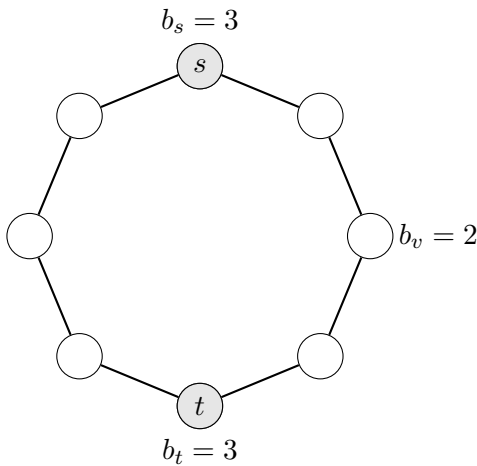

\section{Deferred from Section~\ref{sec:subsettsp}: Bounded-Degree Subset Traveling Salesman Problems}

\subsection{A \texorpdfstring{$(14/9, +6)$}{(14/9, +6)}-Approximation Algorithm for BDSTSP}
\label{app:subset:circuit}
In this appendix, we present a $(14/9, +6)$-approximation algorithm for the bounded-degree subset traveling salesman problem.
Let $(G = (V, E), \{c_e\}_{e \in E}, \{b_v\}_{v \in V}, X)$ be the given instance.

\paragraph*{Algorithm description.}
Solve~\eqref{eq:bd-subset-tsp} to obtain an optimal solution $x^*$.
Let $A := \{v \in V \setminus X : x^*(\delta(v)) < 2 \}$, and let
\[
B_v := \begin{cases}
    1, & v \in A, \\
    \frac{b_v}2, & v \not \in A.
\end{cases}
\]

Apply \Cref{lem:fm-structured} to the bounded-degree Steiner tree instance $(G=(V,E), \{c_e\}_{e \in E},$ $ \{B_v\}_{v \in V}, X)$ with LP solution $\bar x:= \frac{x^*}{2}$ and $A$.
This gives a subgraph $F \subseteq 2G$ spanning $X$, and an inclusion-wise minimal subgraph $F^m \subseteq F$ that is  feasible for the same bounded-degree Steiner tree instance.

Define $T := \odd(F^m)$.
For each vertex $v$, set the ``residual'' degree bound
\[
r_v := b_v + 6 - \deg_{F^m}(v),
\]
and find a minimum-cost bounded-degree $T$-join $J$ of $G$ with degree bounds $r_v$ using \Cref{lem:bdtjoin_integrality}.

Output $H := F^m \uplus J$, after deleting any connected component that contains no terminal.

\paragraph*{Algorithm analysis.}
Let $F^r := F \setminus F^m$ be the multiset of ``remaining'' edges.

\begin{lemma}\label{lem:bdstsp-join-integral}
The residual degree bound $r_v$ is odd if and only if $v \in T$, making $(G, \{c_e\}_{e \in E},$ $ \{r_v \}_{v \in V}, T)$ a valid bounded-degree $T$-join instance.
The vector
\[
y := \frac{1}{9} \chi_{F^m} + \frac{2}{3} \chi_{F^r} + \frac{4}{9} x^*
\]
is a feasible solution to~\eqref{eq:bd-t-join} for this instance.
\end{lemma}

\begin{proof}
First, $r_v$ is odd if and only if $v \in T$ since $b_v$ is even for all $v \in V$.

Now we verify the cut and degree constraints for $y$. Consider a \(T\)-odd cut \(\delta(S)\), i.e., \(|S\cap T|\) is odd.
Since \(|T|\) is even, the complementary side \(V\setminus S\) is also \(T\)-odd.
Thus, by taking the appropriate side of the cut, we can assume without loss of generality that either \(\delta(S)\) is terminal-separating or \(S\subseteq V\setminus X\).

\begin{itemize}

    \item Case 1: $S \cap X \neq \emptyset$, $(V \setminus S) \cap X \neq \emptyset$.
    
    In this case, $x^*(\delta(S)) \ge 2$.
    Also, since $F^m$ connects all terminals, $|\delta_{F^m}(S)| \ge 1$.
    Therefore,
    \[
    y(\delta(S)) \ge \frac{1}{9} \cdot 1 + \frac{4}{9} \cdot 2 = 1.
    \]
    
    \item Case 2: $S \subseteq V \setminus X$ and $S \setminus A \neq \emptyset$.
    
    Let $v \in S \setminus A$.
    Then, $x^*(\delta(S)) \ge x^*(\delta(v)) \ge 2$ due to the third constraints of~\eqref{eq:bd-subset-tsp}.
    Since $S$ is $T$-odd, we have $|S\cap \odd(F^m)| \equiv |S\cap T| \equiv 1 \pmod 2$.
    By the handshaking lemma, $\chi_{F^m}(\delta(S))$ is odd, and hence $\chi_{F^m}(\delta(S)) \ge 1$.
    Thus,
    \[
    y(\delta(S)) \ge \frac{1}{9} \cdot 1 + \frac{4}{9} \cdot 2 = 1.
    \]
    
    \item Case 3: $S \subseteq V \setminus X$ and $S \subseteq A$.
    
    Since $S$ is an $T$-odd cut, $\chi_{F^m}(\delta(S))$ is odd.
    From the minimality of $F^m$, $\chi_{F^m}(\delta(S))$ cannot be 1, analogous to the argument in the proof of \Cref{lem:bdstspp-join-integral}~(Case~5).
    Hence, $\chi_{F^m}(\delta(S)) \ge 3$.
    
    On the other hand, $S \subseteq A$ implies $\chi_{F}(\delta(S))$ is even.
    Since $F = F^m \uplus F^r$ and $\chi_{F^m}(\delta(S))$ is odd, we have $\chi_{F^r}(\delta(S))\ge 1$.
    Consequently,
    \[
    y(\delta(S)) \ge \frac{1}{9} \cdot 3 + \frac{2}{3} \cdot 1 = 1.
    \]
\end{itemize}
This proves all cut constraints.

Now we check the degree constraints.
Since
$
r_v := b_v + 6 - \deg_{F^m}(v)
$, it suffices to verify 
$\deg_{F^m}(v) + y(\delta(v)) \leq b_v+6$ for all $v$.

For every vertex $v$, we have
\[
\deg_{F^m}(v) + y(\delta(v)) = \frac{10}{9} \deg_{F^m}(v) + \frac{2}{3} \deg_{F^r}(v) + \frac{4}{9} x^*(\delta(v)).
\]
If $v \notin A$, $\deg_{F^m}(v) + \deg_{F^r}(v) = \deg_F(v) \le b_v/2 + 3$ and $x^*(\delta(v)) \le b_v$.
Hence,
\[
\deg_{F^m}(v) + y(\delta(v)) \le \frac{10}{9} (\deg_{F^m}(v) + \deg_{F^r}(v)) + \frac{4}{9} b_v \le \frac{10}{9} (b_v/2 + 3) + \frac{4}{9} b_v = b_v + \frac{10}{3} < b_v + 6.
\]

If $v \in A$, then $\deg_{F^m}(v) \le 4$ and $\deg_{F^m}(v) + \deg_{F^r}(v) \le 8$ from the guarantees of \Cref{lem:fm-structured}.
In addition, $b_v \ge 2$ and $x^*(\delta(v)) < 2$.
Therefore,
\begin{align*}
\deg_{F^m}(v) + y(\delta(v)) &< \frac{10}{9} \deg_{F^m}(v) + \frac{2}{3} (8-\deg_{F^m}(v)) + \frac{4}{9} \cdot 2 \\
&= \frac{4}{9} \deg_{F^m}(v) + \frac{56}{9} \\
&\le \frac{4}{9} \cdot 4 + \frac{56}{9} = 8 \\
&\le b_v + 6.
\end{align*}
Hence, $y$ is a feasible solution to~\eqref{eq:bd-t-join}.
\end{proof}

\begingroup
\renewcommand{\thetheorem}{\ref{thm:bd-subset-14-9-6}}
\begin{theorem}
There exists a $(14/9,+6)$-approximation algorithm for the bounded-degree subset traveling salesman problem.
\end{theorem}
\addtocounter{theorem}{-1}
\endgroup

\begin{proof}
Since $F := F^m \uplus F^r$ and $c(F) \le \OPT$, we have $c(F^m) + c(F^r) \le OPT$.
Therefore,
\begin{align*}
c(F^m \uplus J) &\le c(F^m) + c(y) \\
&= \frac{10}{9} c(F^m) + \frac{2}{3} c(F^r) + \frac{4}{9} c(x^*) \\
&\le \frac{10}{9} c(F) + \frac{4}{9} c(x^*) \\
&\le \frac{14}{9} \OPT.
\end{align*}

We can also bound degree violation as follows:
\[
\deg_{F^m\uplus J} (v) \le \deg_{F^m}(v) + r_v = \deg_{F^m}(v) + (b_v + 6 - \deg_{F^m}(v)) = b_v + 6.
\]
Since $F^m$ is connected and spans $X$, the same is true for $H := F^m \uplus J$.
Since $J$ is a $T$-join and $T=\odd(F^m)$, every vertex has even degree in $H$.
\end{proof}

\subsection{A \texorpdfstring{$(2, +8)$}{(2, +8)}-Approximation Algorithm for BDSTSPP}\label{app:bdstspp:2:8}

\paragraph*{Algorithm description.}
Solve~\eqref{eq:bd-subset-path-tsp} to obtain an optimal solution $x^*$. Let  $A := \{v \in V \setminus X : x^*(\delta(v)) < 2\}$, and let
\[
B_v := \begin{cases}
    1, & v\in A, \\
    \frac{b_v + 1}{2}, & v \in \{s, t \}, \\
    \frac{b_v}{2}, & v \in V \setminus (A \cup \{s, t\}).
\end{cases}
\]
Add a zero-cost auxiliary edge \(e_0 = st\) and let $G_0 :=  (V,\,E \cup \{e_0\})$.
Let
\[
\bar x := \frac{x^* + \chi_{\{e_0\}}}{2}.
\]

Apply \Cref{lem:fm-structured} to the bounded-degree Steiner tree instance $(G_0, \{c_e\}_{e \in E}, \{B_v\}_{v \in V}, X)$ with $\bar{x}$ and $A$.
Obtain a connected subgraph $F_0 \subseteq 2G_0$ spanning all terminals such that $c(F_0)\le 2c(\bar x)$, $\deg_{F_0}(v)\le B_v+3$ for all \(v\notin A\), and \(\deg_{F_0}(v)\le 8\) with \(\deg_{F_0}(v)\) even for all \(v\in A\).

Let \(F\) be the subgraph of the original graph \(G\) obtained by deleting all copies of \(e_0\) from \(F_0\).
Define \(T := \odd(F)\triangle\{s,t\}\).
For each vertex \(v\), set the residual degree bound
\[
r_v := b_v + 8 - \deg_F(v),
\]
and find a minimum-cost bounded-degree $T$-join $J$ of $G$ with degree bounds $r_v$ using \Cref{lem:bdtjoin_integrality}.

Output \(H:=F\uplus J\), after deleting any connected component that contains no terminal.

\paragraph*{Algorithm analysis.}
Let \(P\) be a shortest simple \(s\)-\(t\) path in \(G\).
Since every feasible BDSTSPP solution contains an \(s\)-\(t\) path as a subgraph, \(c(P)\le\OPT\).
Analogous to the observation at the beginning of \Cref{sec:511},  \(F\) either has a connected component containing all terminals, or  has two terminal-containing components: one containing \(s\) and the other containing \(t\).

\begin{lemma}\label{lem:bdstspp-2-8-join-integral}
The residual degree bound \(r_v\) is odd if and only if \(v\in T\), making $(G,\{c_e\}_{e \in E},$ $\{r_v\}_{v\in V},T)$ a valid bounded-degree \(T\)-join instance.
The vector
\[
y:=\frac12x^*+\frac12\chi_P
\]
is feasible to~\eqref{eq:bd-t-join} for this instance.
\end{lemma}

\begin{proof}
We check the parity condition for the degree bounds and then verify the cut and degree inequalities for \(y\).

First, \(r_v\) is odd if and only if \(v\in T\).
For \(v\in\{s,t\}\), the bound \(b_v\) is odd, so \(r_v=b_v+8-\deg_F(v)\) is odd if and only if \(\deg_F(v)\) is even.
This is exactly the condition \(v\in\odd(F)\triangle\{s,t\}\).
For \(v\notin\{s,t\}\), the bound \(b_v\) is even, so \(r_v\) is odd if and only if \(\deg_F(v)\) is odd, again exactly when \(v\in T\).

Now we verify the cut and degree inequalities for \(y\). Fix a \(T\)-odd cut \(\delta(S)\), i.e., \(|S\cap T|\) is odd.
Since \(|T|\) is even, the complementary side \(V\setminus S\) is also \(T\)-odd.
Thus, by taking the appropriate side of the cut, we can assume without loss of generality that either \(\delta(S)\) is terminal-separating or \(S\subseteq V\setminus X\).

\begin{itemize}
    \item Case 1: \(S\cap X\neq\emptyset\) and \((V\setminus S)\cap X\neq\emptyset\), and \(\delta(S)\) is an \(s\)-\(t\) cut.

    In this case, \(x^*(\delta(S))\ge1\), and \(P\) crosses \(\delta(S)\).
    Hence
    \[
    y(\delta(S))\ge \frac12\cdot1+\frac12\cdot1=1.
    \]

    \item Case 2: \(S\cap X\neq\emptyset\) and \((V\setminus S)\cap X\neq\emptyset\), and \(\delta(S)\) is not an \(s\)-\(t\) cut.

    In this case, \(x^*(\delta(S))\ge2\).
    Hence
    \[
    y(\delta(S))\ge \frac12\cdot2=1.
    \]

    \item Case 3: \(S\subseteq V\setminus X\).

    Since every vertex of \(A\) has even degree in \(F_0\), and \(e_0\) is not incident to any vertex of \(A\), every vertex of \(A\) has even degree in \(F\).
    Hence, \(A\cap T=\emptyset\) since $s,t\notin A$.
    This shows that  \(S\) cannot be contained in \(A\), since \(S\) is \(T\)-odd.
    Pick \(v\in S\setminus A\).
    By the fourth constraint of~\eqref{eq:bd-subset-path-tsp}, \(x^*(\delta(S))\ge x^*(\delta(v))\ge2\).
    Hence
    \[
    y(\delta(S))\ge \frac12\cdot2=1.
    \]
\end{itemize}
This proves all \(T\)-odd cut inequalities.

It remains to check the degree inequalities.
Since $r_v := b_v + 8 - \deg_F(v)$, it suffices to verify $\deg_F(v) + y(\delta(v)) \le b_v + 8$ holds for all $v$.

First consider \(v\in V\setminus(A\cup\{s,t\})\).
Then \(\deg_F(v)\le \frac{b_v}{2}+3\), \(x^*(\delta(v))\le b_v\), and \(\deg_P(v)\le2\).
Therefore,
\[
\deg_F(v)+y(\delta(v))\le \left(\frac{b_v}{2}+3\right)+\frac12 b_v+\frac12\cdot2=b_v+4< b_v+8.
\]

Next consider \(v\in A\). Note that $b_v \ge 2$. We have \(\deg_F(v)\le8\), \(x^*(\delta(v))<2\), and \(\deg_P(v)\le2\).
Therefore,
\[
\deg_F(v)+y(\delta(v))<8+\frac12\cdot2+\frac12\cdot2=10\le b_v+8.
\]

Finally, consider \(v\in\{s,t\}\). We have \(\deg_F(v)\le \frac{b_v+1}{2}+3\), \(x^*(\delta(v))\le b_v\), and \(\deg_P(v)=1\).
Therefore,
\[
\deg_F(v)+y(\delta(v))\le \left(\frac{b_v+1}{2}+3\right)+\frac12 b_v+\frac12=b_v+4< b_v+8.
\]

This shows that \(y\) is feasible solution to \eqref{eq:bd-t-join}.
\end{proof}

\begingroup
\renewcommand{\thetheorem}{\ref{thm:bd-subset-path-2-8}}
\begin{theorem}
There exists a $(2, +8)$-approximation algorithm for the bounded-degree subset traveling salesman path problem.
\end{theorem}
\addtocounter{theorem}{-1}
\endgroup

\begin{proof}
The same argument as in the proof of Theorem~\ref{thm:bd-subset-path-11-5-6} shows that \(\bar x\) is a feasible solution to~\eqref{eq:bdst} for \(G_0\) with the set of terminals \(X\) and degree bounds \(\{B_v\}_{v \in V}\), and therefore
\[
c(F)=c(F_0)\le 2c(\bar x)=c(x^*)\le\OPT.
\]
\Cref{lem:bdstspp-2-8-join-integral} implies that the \(T\)-join \(J\) found by the algorithm satisfies \(c(J)\le c(y)\) and \(\deg_J(v)\le r_v\) for all \(v\in V\).
Thus,
\[
c(H)\le c(F)+c(J)\le c(F)+c(y)\le \OPT+\frac12c(x^*)+\frac12c(P)\le2\OPT.
\]

The degree bound follows from the definition of \(r_v\): $\deg_H(v)\le \deg_F(v)+\deg_J(v)\le \deg_F(v)+r_v=b_v+8$.

Again, the same argument as in the proof of Theorem~\ref{thm:bd-subset-path-11-5-6} shows that the algorithm outputs a feasible solution, completing the proof.
\end{proof}

\end{document}